\documentclass[11pt]{article}
\usepackage{amsmath,amsthm,amssymb,color,verbatim,graphicx}
\newcommand{\remove}[1]{}
\newtheorem{thm}{Theorem}[section]
\newtheorem{clm}[thm]{Claim}
\newtheorem{lem}[thm]{Lemma}
\newtheorem{define}[thm]{Definition}
\newtheorem{cor}[thm]{Corollary}

\newtheorem{conjecture}{Conjecture}
\newtheorem{THM}{Theorem}

\newtheorem{remark}[thm]{Remark}

\renewcommand{\remove}[1]{}

\newcommand{\eps}{{\varepsilon}}
\renewcommand{\l}{\left}
\renewcommand{\r}{\right}

\newcommand{\de}{{\delta}}

\newcommand{\comments}[1]{}
\newcommand{\rank}{\textnormal{rank}}

\newcommand{\ML}{\textnormal{ML}}
\renewcommand{\deg}{\textnormal{deg}}
\newcommand{\LDR}{\textnormal{LDR}}
\newcommand{\J}{\textnormal{J}}
\newcommand{\depth}{\textnormal{depth}}
\newcommand{\dist}{\textnormal{dist}}

\newcommand{\RM}{\textnormal{RM}}

\newcommand{\Poly}{\textnormal{Poly}}

\newcommand{\restate}[2]{\medskip \noindent {\bf #1.} {\sl #2}}

\def\F{{\mathbb{F}}}

\newcommand{\R}{\mathbb{R}}
\newcommand{\N}{\mathbb{N}}
\newcommand{\T}{\mathbb{T}}
\newcommand{\E}{\mathbb{E}}
\newcommand{\Z}{\mathbb{Z}}
\newcommand{\U}{\mathbb{U}}

\newcommand{\calP}{\mathcal{P}}
\newcommand{\calL}{\mathcal{L}}

\newcommand{\calH}{\mathcal{H}}
\newcommand{\B}{\mathcal{B}}

\newcommand{\C}{\mathbb{C}}

\renewcommand{\Pr}{\mathbf{Pr}}

\def\draft{0}   

\ifnum\draft=1 
    \def\ShowAuthNotes{1}
\else
    \def\ShowAuthNotes{0}
\fi

\ifnum\ShowAuthNotes=0
\newcommand{\authnote}[2]{{ \footnotesize \bf{\color{red}[#1's Note: {\color{blue}#2}]}}}
\else
\newcommand{\authnote}[2]{}
\fi

\begin{document}
\title{List decoding Reed-Muller codes over small fields}

\author{
Abhishek Bhowmick\thanks{Research supported in part by NSF Grant CCF-1218723.}\\
Department of Computer Science\\
The University of Texas at Austin\\
\texttt{bhowmick@cs.utexas.edu}
\and
Shachar Lovett \thanks{Supported by NSF CAREER award 1350481}\\
Department of Computer Science and Engineering\\
University of California, San Diego\\
\texttt{slovett@ucsd.edu}}

\maketitle
\thispagestyle{empty}
\begin{abstract}
The list decoding problem for a code asks for the maximal radius up to which any ball of that radius contains only a constant number of codewords. The list decoding radius is not well understood even for well studied codes, like Reed-Solomon or Reed-Muller codes.

Fix a finite field $\F$. The Reed-Muller code $\RM_{\F}(n,d)$ is defined by $n$-variate degree-$d$ polynomials over $\F$. In this work, we study the list decoding radius of Reed-Muller codes over a constant prime field $\F=\F_p$, constant degree $d$ and large $n$. We show that the list decoding radius is equal to the minimal distance of the code.

That is, if we denote by $\de(d)$ the normalized minimal distance of $\RM_{\F}(n,d)$, then the number of codewords in any ball of radius $\de(d)-\eps$ is bounded by $c=c(p,d,\eps)$ independent of $n$. This resolves a conjecture of Gopalan-Klivans-Zuckerman [STOC 2008],  who among other results proved it in the special case of $\F=\F_2$; and extends the work of Gopalan [FOCS 2010] who proved the conjecture in the case of $d=2$.

We also analyse the number of codewords in balls of radius exceeding the minimal distance of the code. For $e \leq d$, we show that the number of codewords of $\RM_{\F}(n,d)$ in a ball of radius $\de(e) - \eps$ is bounded by $\exp(c  \cdot  n^{d-e})$, where $c=c(p,d,\eps)$ is independent of $n$. The dependence on $n$ is tight. This extends the work of Kaufman-Lovett-Porat [IEEE Inf. Theory 2012] who proved similar bounds over $\F_2$.

The proof relies on several new ingredients: an extension of the Frieze-Kannan weak regularity to general function spaces, higher-order Fourier analysis, and an extension of the Schwartz-Zippel lemma to compositions of polynomials.
\end{abstract}

\newpage
\setcounter{page}{1}

\section{Introduction}
The concept of \emph{list decoding} was introduced by Elias \cite{Elias} and Wozencraft \cite{Woz} to decode \emph{error correcting codes} beyond half the minimum distance. The objective of list decoding is to output all the codewords within a specified radius around the received word. After the seminal results of Goldreich and Levin \cite{GL} and Sudan \cite{Sudan} which gave list decoding algorithms for the Hadamard code and the Reed-Solomon code respectively, there has been tremendous progress in designing list decodable codes. See the excellent surveys of Guruswami \cite{Venkat:book, Venkat:thesis} and Sudan \cite{Sudan:survey}.

List decoding has applications in many areas of computer science including hardness amplification in complexity theory \cite{STV, luca-xor}, derandomization \cite{Vad}, construction of hard core predicates from one way functions \cite{GL, AGS}, construction of extractors and pseudorandom generators \cite{TZS, SU} and computational learning \cite{KM,Jackson}. Despite so much progress, the largest radius up to which list decoding is tractable is still a fundamental open problem even for well studied codes like Reed-Solomon (univariate polynomials) and Reed-Muller codes (multivariate polynomials). The goal of this work is to analyse Reed-Muller codes over small fields and small degree.

Reed-Muller codes (RM codes) were discovered by Muller in 1954. Fix a finite field $\F=\F_q$. Let $d \in \N$. The RM code $\RM_{\F}(n,d)$ is defined as follows. The message space consists of degree $\leq d$ polynomials in $n$ variables over $\F$ and the codewords are evaluation of these polynomials on $\F^n$. Let $\de_{p}(d)$ denote the normalized distance of $\RM_{\F}(n,d)$. Let $d=a(q-1)+b$ where $0 \leq b<q-1$. We have
$$
\de_{\F}(d)=\frac{1}{q^a}\l(1-\frac{b}{q}\r).
$$

RM codes are one of the most well studied error correcting codes. Many of the applications in computer science involves low degree polynomials over small fields, namely RM codes. Given a received word $g:\F^n \rightarrow \F$ the objective is to output the list of codewords (e.g. low-degree polynomials) that lie within some distance of $g$. Typically we will be interested in regimes where list size is either independent of $n$ or polynomial in the block length $\F^n$.

\subsection{Previous Work} Let $\calP_{d}(\F^n)$ denote the class of degree $\leq d$ polynomials $f:\F^n \rightarrow \F$. Let $\dist$ denote the normalized Hamming distance.
For $\RM_{\F}(n,d)$, $\eta>0$, let
$$
\ell_{\F}(n,d,\eta):=\max_{g:\F^n \rightarrow \F}\l|\{f \in \calP_d(\F^n):\dist(f,g) \leq \eta\}\r|.
$$
Let $\LDR_{\F}(n,d)$ (short for \emph{list decoding radius}) be the maximum $\eta$ for which $\ell_{\F}(n,d,\eta-\eps)$ is upper bounded by a constant depending only on $\eps, |\F|,d$ for all $\eps>0$.

It is easy to see that $\LDR_{\F}(n,d)\leq \de_{\F}(d)$. The difficulty lies in proving a matching lower bound. The first breakthrough result was in the setting of $d=1$ over $\F_2$ (Hadamard Codes) where Goldreich and Levin showed that $\LDR_{\F_2}(n,1)=\de_{\F_2}(1)=1/2$ \cite{GL}. Later, Goldreich, Rubinfield and Sudan \cite{GRS} generalized the field to obtain $\LDR_{\F}(n,1)=\de_{\F}(1)=1-1/|\F|$. In the setting of $d<|\F|$, Sudan, Trevisan and Vadhan \cite{STV} showed that $\LDR_{\F}(n,d) \ge 1-\sqrt{2d/|\F|}$ improving previous work by Arora and Sudan \cite{AroraSudan}, Goldreich \emph{et al} \cite{GRS} and Pellikaan and Wu \cite{PellikaanWu}. A crucial result that was a bulding block in the multivariate setting was the problem of list decoding Reed-Solomon codes which was analysed by Sudan \cite{Sudan} and Guruswami and Sudan \cite{GuruswamiSudan}.
The list decoding radius obtained above essentially attains the Johnson radius, which is a radius such that
for any code over $\F$ with normalized minimum distance $\de$, the list decoding radius ($\LDR$) is at least
$$
\J_{\F}(\de):=\left(1-\frac{1}{|\F|}\right)\left(1-\sqrt{1-\frac{|\F|\de}{|\F|-1}}\right).
$$
There have been few results that show list decodability beyond the Johnson radius \cite{DGKS08, GKZ08}.

In 2008, Gopalan, Klivans and Zuckerman \cite{GKZ08} showed that $\LDR_{\F_2}(n,d)=\de_{\F_2}(d)$. This beats the Johnson radius already for $d \geq 2$. The list decoding algorithm in \cite{GKZ08} is a generalization of the Goldreich-Levin algorithm \cite{GL}. However their algorithm crucially depends on the fact that the ratio of minimum distance to unique decoding radius is equal to $2$ which is the size of the field. Therefore, it does not generalize to higher fields (except for some special cases). They pose the following conjecture.
\begin{conjecture}[\cite{GKZ08}]\label{conj:1}For all constants $d$ and all fields $\F$, $\LDR_{\F}(n,d)=\de_{\F}(d)$.
\end{conjecture}

An important contribution of \cite{GKZ08} is an algorithm for list decoding that outputs the list of codewords up to radius $\eta$ efficiently assuming $\ell_{\F}(n,d,\eta)$ is bounded.

It was also shown \cite{GKZ08} that $\LDR_{\F}(n,d) \geq \frac{1}{2}\de_{\F}(d-1)$ and this beats the Johnson radius already when $d$ is large. It is believed \cite{GKZ08, Gopalan10} that the hardest case is the setting of small $d$.
An important step in this direction was taken in \cite{Gopalan10} that considered quadratic polynomials and showed that $\LDR_{\F}(n,2) = \de_{\F}(2)$ for all fields $\F$ and thus proved the conjecture for $d=2$. In the setting of $\F_2$, Kaufman, Lovett and Porat \cite{KLP10} showed tight list sizes for radii beyond the minimum distance.

\subsection{Our Results}
As mentioned before, the algorithmic problem of list decoding was reduced to the combinatorial problem in \cite{GKZ08}. Our main theorem is a resolution of Conjecture~\ref{conj:1} for prime fields. We note that prior to this, the conjecture was open even in the $d<|\F|$ case.

\begin{THM}\label{THM:main}Let $\F=\F_p$ be a prime field. Let $\eps>0$ and $d,n \in \N$. Then, $$\ell_{\F}(d,n,\de_{\F}(d)-\eps) \leq c_{p,d,\eps}.$$
\end{THM}

\begin{remark}[Algorithmic Implications]As mentioned above, using the reduction of algorithmic list decoding to combinatorial list decoding in \cite{GKZ08} along with Theorem~\ref{THM:main}, for fixed prime fields, $d$ and $\eps>0$, we now have list decoding algorithms in both the global setting (running time polynomial in $|\F|^n$) and the local setting (running time polynomial in $n^d$).
\end{remark}

Next, we study list sizes for radii which are larger than the minimal radius of the code. We give bounds which capture the correct
exponent of $n$ for all radii. This extends the results of Kaufman, Lovett and Porat \cite{KLP10} who studied Reed-Muller codes over $\F_2$, to all prime fields.

\begin{THM}\label{THM:main2}Let $\F=\F_p$ be a prime field. Let $\eps>0$ and $e \leq d,n \in \N$. Then, $$\ell_{\F}(d,n,\de_{\F}(e)-\eps) \leq \exp\l(c_{p,d,\eps}n^{d-e}\r)$$
\end{THM}

\begin{remark}
The exponent of $n$ in Theorem \ref{THM:main2} is tight, as the following example shows. Let $e=a(p-1)+b$ with $0 \le b < p-1$. Consider polynomials of the form
$$
P(x)=\l(\prod_{i=i}^a(x_i^{p-1}-1)\r)\l(\prod_{j=1}^b(x_{a+1}-j)\r)\left(x_{a+2}+Q(x_{a+3},\ldots,x_n)\right)
$$
for all polynomials $Q$ of degree $d-e$. Observe that $\Pr[P(x) \neq 0]=  \frac{1}{p^a}\l(1-\frac{b}{p}\r) \l(1-\frac{1}{p}\r) = \de(e)(1-1/p)$. The number of such polynomials is $\exp(c' n^{d-e})$ for
some $c'=c'_{p,d,e}$.
\end{remark}

\subsection{Proof overview}

Previous results have mostly relied on the idea of local correction of the RM code. The work of \cite{Gopalan10} uses (linear) Fourier analysis which does not seem to go beyond quadratic polynomials. We use tools from higher order Fourier analysis to resolve the conjecture. We think of $\F=\F_p,d,\eps$ as constants. For a received word $g:\F^n \to \F$ our goal is to upper bound $\l|\{f \in \calP_d(\F^n):\dist(f,g) \leq \eta\}\r|$. For simplicity of exposition, we assume in the proof overview that $d<|\F|$. The general case is somewhat more technical, as it requires the introduction of nonclassical polynomials.

\paragraph{A weak regularity (A low complexity proxy for the received word).}
The first step is an extension of the Frieze-Kannan weak regularity \cite{FK99} which would allow us to move from an arbitrary received word $g$ to a "low complexity" received word. We note that a somewhat similar idea appeared also in \cite{trevisan2009regularity}.

Let $X,Y$ be finite sets and let $P(Y):=\{f:Y \rightarrow \R_{\geq 0}  :\sum_{y \in Y}f(y)=1\}$ be the probability simplex over $Y$. We view functions $f:X \to P(Y)$ as randomized functions from $X$ to $Y$. For $f,g:X \rightarrow P(Y)$ we define $$\Pr_{x}[f(x)=g(x)]:=\E_x\langle f(x),g(x)\rangle.$$ Given $\eps>0$, any function $g:X \rightarrow P(Y)$ and a collection $F$ of functions $f:X \rightarrow P(Y)$, one can find a collection of $c:=1/\eps^2$ functions $h_1, \ldots , h_c \in F$ and a \emph{proxy} $g_1:X \rightarrow P(Y)$ for $g$, such that $g_1$ is determined by $h_1(x),\ldots , h_c(x)$ and such that $g_1$ is indistinguishable from $g$ with respect to $F$.

\restate{Lemma \ref{lem:pseudorandom}}
{Let $g:X \rightarrow P(Y)$, $\eps>0$, and $F$ be a collection of functions $f:X \rightarrow P(Y)$. Then there exist $c \leq 1/\eps^2$ functions $h_1, h_2,\ldots , h_c \in F$ and a function $\Gamma:P(Y)^c \rightarrow P(Y)$ such that for all $f \in F$, $$\l|\Pr[g(x)=f(x)]-\Pr[\Gamma(h_1(x),h_2(x),\ldots , h_c(x))=f(x)]\r| \leq \eps.$$}

In our case, $X=\F^n$, $Y=\F$ and $F=\calP_d(\F^n)$.
When $F$ is a family of "deterministic" functions $f:X \to Y$, as it is in our case, we can obtain one-sided approximation using only deterministic functions $h_1,\ldots,h_c$.

\restate{Corollary \ref{cor:pseudorandom}}
{Let $g:X \rightarrow Y$, $\eps>0$, and $F$ be a collection of functions $f:X \rightarrow Y$. Then there exist $c \leq 1/\eps^2$ functions $h_1, h_2,\ldots , h_c \in F$ such that for every $f \in F$, there is a function $\Gamma_f:Y^c \rightarrow Y$ such that $$\Pr_x[\Gamma_f(h_1(x), \ldots , h_c(x))=f(x)] \geq \Pr_x[g(x)=f(x)] - \eps.$$}

\paragraph{Strong regularity applied to $\calH$.}  The collection of polynomials $\calH=\{h_1,\ldots , h_c\} \subset \calP_d(\F^n)$ defines a partition of the input space $\F^n$ into \emph{atoms} $\{x \in \F^n: h_1(x)=a_1,\ldots,h_c(x)=a_c\}$. We next regularize $\calH$. The objective of regularization is to further refine the partition into smaller atoms with the goal that the polynomials $h_1,\ldots,h_c$ are "pseudo-random". Formally, we require the polynomials to be inapproximable by lower degree polynomials, which is equivalent to having negligible Gowers uniformity norm. This ensures, for example, that for uniformly random $X$ in $\F^n$, the distribution $(h_1(X),\ldots , h_c(X))$ is close to uniform over the atoms. This process of regularization was introduced by \cite{GT09} and is now standard in higher-order Fourier analysis.
Let $\calH'=\{h'_1,\ldots,h'_{c'}\} \subset \calP_d(\F^n)$ be the regularized $\calH$ that satisfies the above properties, where $c'=c'(p,d,c)$.

\paragraph{Structure of polynomials close to low complexity received words.}
Fix now an $f \in \calP_d(\F^n)$ such that $\dist(f,g) \leq \de_p(d)-\eps$. We will show that $f$ must be determined by $\cal H'$. That is,
$$
f(x) = F(h'_1(x),\ldots,h'_{c'}(x))
$$
for some $F:\F^{c'} \to \F$. This will bound the number of such functions by $p^{p^{c'}}$, which is independent of $n$.

In order to achieve that, we regularize the family of polynomials $\calH' \cup \{f\}$. By choosing regularity parameters appropriately, we can assure that only $f$ decomposes further,
$$
f = F(h'_1(x),\ldots,h'_{c'}(x),h''_1(x),\ldots,h''_{c''}(x))
$$
where $\calH'' = \{h_1,\ldots,h'_{c'},h''_1,\ldots,h''_{c''}\}$ is regular. Moreover, for $G_f(h'_1(x),\ldots,h'_{c'}(x))=\Gamma_f(h_1(x),\ldots,h_c(x))$, we know that
$$
\Pr[f(x) = G_f(h'_1(x),\ldots,h'_{c'}(x))] \ge 1-\de_p(d)+\eps/2.
$$
The regularity of $\calH''$ allows us to reduce the question to that of the structure of $F$ vs $G_f$. We then show, by a variant of the Schwartz-Zippel lemma, that such an approximation can only exist when $F$ does not depend on $h''_1,\ldots,h''_{c''}$. The bound for larger radii $\de_{\F}(e)-\eps$ with $e<d$ follows along similar lines. We show that in the decomposition above, since $\Pr[F=G_f] > 1-\de_{\F}(e)+\eps/2$, this can only occur when $h''_1,\ldots,h''_{c''}$ have degree at most $d-e$. As the number of such polynomials is exponential in $n^{d-e}$, we derive similar bounds for the number of functions $f$.

\section{Preliminaries}
\subsection{Notation}
Let $\N$ denote the set of positive integers. For $n \in \N$, let $[n]:=\{1,2,\ldots , n\}$. We use $y=x \pm \eps$ to denote $y \in [x-\eps, x+\eps]$. Let $\T$ denote the torus $\R/\Z$. This is an abelian group under addition. For $n \in \N$, and $x,y \in \C^n$, let $\langle x,y \rangle:=\sum_{i=1}^n x_i\overline{y_i}$ where $\overline{a}$ is the conjugate of $a$. Let $||x||_2:=\sqrt{\langle x,x \rangle}$.

Fix a prime field $\F=\F_p$. Let $|.|$ denote the natural map from $\F$ to $\{0,1,\ldots , p-1\}\in \Z$. Let $e:\T \rightarrow \C$ be the map $e(x):=e^{2\pi i x}$. Let $e_p:\F \rightarrow \C$ be the map $e_p(x)=e(\frac{|x|}{p})$. For an integer $k \geq 0$, let $\U_{k}:=\frac{1}{p^k}\Z/\Z$. Note that $\U_k$ is a subgroup of $\T$. Let $\iota:\F \rightarrow \U_1$ be the bijection $\iota(a)=\frac{|a|}{p} \pmod 1$.

For a finite set $X$ and $n \in \N$, with $f:X\rightarrow \C^n$, we write $\E_{x}f(x)$ to denote $\frac{1}{|X|}\sum_{x \in X}f(x)$. We define $||f||_2:=\sqrt{\E_x ||f(x)||_2^2}$. If $g: X \rightarrow \C^n$, we have $\langle f,g \rangle  :=  \E_x \langle f(x),g(x) \rangle$. Let $Y$ be a finite set. Let $P(Y):=\{f:Y \rightarrow \R_{\geq 0}  :\sum_{y \in Y}f(y)=1\}$ denote the probability simplex on $Y$. We shall write randomized functions by mapping them to the simplex. Thus, for $f,g:X \rightarrow P(Y)$ we define $$\Pr_{x}[f(x)=g(x)]:=\E_x\langle f(x),g(x)\rangle.$$
If $f:X \rightarrow Y$ is a deterministic function, then we embed $Y$ into $P(Y)$ in the obvious way, and consider $f:X \to P(Y)$ with $f(x)_{y}=1$ if $f(x)=y$ when viewed as a function to $Y$, and $f(x)_{y'}=0$ for all $y' \in Y \setminus \{y\}$.

\subsection{Polynomials}

\begin{define}[Derivative] Given a function $f:\F^n \rightarrow \T$ and $a \in \F^n$, define the derivative of $f$ in direction $a$ as $D_af:\F^n \rightarrow \T$ as $D_af(x)=f(x+a)-f(x)$ for $x \in \F^n$.
\end{define}

\begin{define}[Nonclassical Polynomial or Polynomial]\label{def:poly}Let $d \in \N$. Then $f:\F^n \rightarrow \T$ is a polynomial of degree $\leq d$ if for all $a_1,\ldots , a_{d+1},x \in \F^n$, \begin{equation}\label{eq:der}\l(D_{a_1}\ldots D_{a_{d+1}}f\r)\l(x\r)=0.\end{equation}
\end{define}

The degree of $f$ denoted by $\deg(f)$ is the smallest such $d \in \N$ for which the above holds. If the image of $f$ lies in $\U_1$ then $f$ is called a classical polynomial of degree $d$. When $d<|\F|$, it is known that all the polynomials of degree $d$ satisfying \eqref{eq:der} are classical polynomials. However, when $d \geq |\F|$, there exist nonclassical polynomials. We write $\Poly_{\leq d}(\F^n \rightarrow \T)$ to denote the class of degree $ \leq d$ polynomials. Unless explicitly specified, a polynomial is a (potentially) nonclassical polynomial. The following lemma from \cite{TZ} characterizes polynomials.

\begin{lem}[\cite{TZ}, Lemma 1.7]\label{lem:poly}  Let $d \in \N$.
\begin{itemize}
\item A function $f:\F^n \rightarrow \T$ is a polynomial of degree $\leq d$ if and only if $D_af$ is a polynomial of degree $\leq d-1$ for all $a\in \F^n$.
\item A function $f:\F^n \rightarrow \T$ is a classical polynomial with $deg(f) \leq d$ if $f=\iota \circ P$ where $P:\F^n \rightarrow \F$ is of the form $$P(x_1,\ldots , x_n)=\sum_{0 \leq d_1,\ldots , d_n \leq p-1:\\ \sum_{i}d_i \leq d}c_{d_1,\ldots , d_n}\prod_{i=1}^nx_i^{d_i},$$ where $c_{d_1,\ldots , d_n} \in \F$ are unique.
\item A function $f:\F^n \rightarrow \T$ is a polynomial with $deg(f) \leq d$ if $f$ is of the form $$f(x_1,\ldots , x_n)=\alpha+\sum_{0 \leq d_1,\ldots , d_n \leq p-1, k \geq 0:\\ \sum_{i}d_i \leq d-k(p-1)}\frac{c_{d_1,\ldots , d_n,k}\prod_{i=1}^n|x_i|^{d_i}}{p^{k+1}} \pmod 1,$$ where $c_{d_1,\ldots , d_n,k} \in \{0,\ldots , p-1\}$ and $\alpha \in \T$ are unique. $\alpha$ is called the shift of $f$ and the largest $k$ such that some $c_{d_1,\ldots , d_n,k}\neq 0$ is the depth of $f$, denoted by $\depth(f)$. Note that classical polynomials have $0$ shift and $0$ depth.
\item If $f:\F^n \rightarrow \T$ is a polynomial with $\depth(f)=k$, then its image lies in a coset of $\U_{k+1}$.
\item If $f:\F^n \rightarrow \T$ is a polynomial such that $\deg(f)=d$ and $\depth(f)=k$, then $\deg(pf)=\max (d-p+1,0)$ and $\depth(pf)=k-1$. Also, if $c \in \{1,\ldots,p-1\}$ then the degree and depth of $cf$ remain unchanged.
\end{itemize}
\end{lem}

Throughout the article, we assume without loss of generality that nonclassical polynomials have zero shift.

\subsection{Rank and Polynomial Factors}
\begin{define}[Rank] Let $d \in \N$ and $f:\F^n \rightarrow \T$. Then $\rank_d(f)$ is defined as the smallest integer $r$ such that there exist polynomials $h_1,\ldots , h_r:\F^n \rightarrow \T$ of degree $\leq d-1$ and a function $\Gamma:\T^r \rightarrow \T$ such that $f(x)=\Gamma(h_1(x),\ldots , h_r(x))$. If $d=1$, then the rank is $0$ if $f$ is a constant function and is
$\infty$ otherwise. If $f$ is a polynomial, then $\rank(f)=\rank_d(f)$ where $d=\deg(f)$.
\end{define}

\begin{define}[Factor] Let $X$ be a finite set. Then a factor $\B$ is a partition of the set $X$. The subsets in the partition are called atoms.
\end{define}

For sets $X$ and $Y$, and a factor $\B$ of $X$, a function $f:X\rightarrow P(Y)$ is said to be measurable with respect to $\B$ if it is constant on the atoms of $\B$.
The average of $f$ over $\B$ is $\E[f|\B]:X \to P(Y)$ defined as
$$
\E[f|\B](x)=\E_{y \in \B(x)}[f(y)]
$$
where $\B(x)$ is the atom containing $x$. Clearly, $\E[f|\B]$ is measurable with respect to $\B$.

A collection of functions $h_1,\ldots , h_c:X \rightarrow Y$ defines a factor $\B$ whose atoms are $\{x \in X:h_1(x)=y_1,\ldots,h_c(x)=y_c\}$ for every $(y_1,\ldots , y_c) \in Y^c$. We use $\B$ to also denote the map $x \mapsto \l(h_1(x),\ldots , h_c(x)\r)$. A function $f$ is measurable with respect to a collection of functions if it is measurable with respect to the factor the collection defines.

\begin{define}[Polynomial Factor] A polynomial factor $\B$ is a factor defined by a collection of polynomials $\calH=\{h_1,\ldots , h_c:\F^n \rightarrow \T\}$ and the factor is written as $\B_{\calH}$. The degree of the factor is the maximum degree of $h \in \calH$. \end{define}

Let $|\B|$ be the number of polynomials defining the factor. If $\depth(h_i)=k_i$ above, then we define $||\B||:=\prod_{i=1}^c p^{k_i+1}$ to be the number of (possibly empty) atoms.

\begin{define}[Rank and Regularity of Polynomial Factor]Let $\B$ be a polynomial factor defined by $h_1,\ldots , h_c:\F^n \rightarrow \T$ such that $\depth(h_i)=k_i$ for $i \in [c]$. Then, the rank of $\B$ is the least integer $r$ such that there exists $(a_1,\ldots , a_c) \in \Z^c$, $\l(a_1 \mod p^{k_1+1}, \ldots , a_c \mod p^{k_c+1}\r) \neq (0,\ldots , 0)$ for which the linear combination $h(x):=\sum_{i=1}^c  a_i h_i(x)$ has $\rank_d(h) \leq r$ where $d=\max_{i}\deg(a_i h_i)$. For a non decreasing function $r:\N \rightarrow \N$, a factor $\B$ is $r$-regular if its rank is at least $r(|\B|)$.
\end{define}

\begin{define}[Semantic and Syntactic refinement]Let $\B$ and $\B'$ be polynomial factors on $\F^n$. A factor $\B'$ is a syntactic refinement of $\B$, denoted by $\B' \succeq_{syn}\B$ if the set of polynomials defining $\B$ is a subset of the set of polynomials defining $\B'$. It is a semantic refinement, denoted by $\B' \succeq_{sem}\B$ if for every $x,y \in \F^n$, $\B'(x)=\B'(y)$ implies $\B(x)=\B(y)$.
\end{define}

We will use the following regularity lemma proved in \cite{BFHHL}.
\begin{lem}[Polynomial Regularity Lemma \cite{BFHHL}]\label{lem:reg}  Let $r:\N \rightarrow \N$ be a non-decreasing function and $d \in \N$. Then there is a function $C_{r,d}^{(\ref{lem:reg})}:\N \rightarrow \N$ such that the following is true. Let $\B$ be a factor defined by polynomials $P_1,\dots , P_c:\F^n \rightarrow \T$ of degree at most $d$. Then, there is an $r$-regular factor $\B'$ defined by polynomials $Q_1,\ldots , Q_{c'}:\F^n \rightarrow \T$ of degree at most $d$ such that $\B' \succeq_{sem} \B$ and $c' \leq C_{r,d}^{(\ref{lem:reg})}(c)$.

Moreover if $\B \succeq_{sem} \hat{\B}$ for some polynomial factor $\hat{\B}$ that has rank at least $r(c')+c'+1$, then $\B' \succeq_{syn} \hat{B}$.
\end{lem}

The next lemma shows that a regular factor has atoms of roughly equal size.
\begin{lem}[Size of atoms \cite{BFHHL}]\label{lem:atomsize}Given $\eps>0$, let $\B$ be a polynomial factor of rank at least $r_{d}^{(\ref{lem:atomsize})}(\eps)$ defined by polynomials $P_1,\ldots , P_c:\F^n \rightarrow \T$ of degree at most $d$ such that $\depth(P_i)=k_i$ for $i \in [c]$. For every $b \in \otimes_{i=1}^c  \U_{k_i+1}$, $$ \Pr_x[\B(x)=b]=\frac{1}{||\B||} \pm \eps.$$
\end{lem}

Finally, we shall need the following lemma which shows that a function of high rank polynomials has the degree one expects.
\begin{lem}[Preserving degree \cite{BFHHL}]\label{lem:degree}  Let $d>0$ be an integer and let $P_1,\ldots , P_c:\F^n \rightarrow \T$ be polynomials of degree at most $d$ that form a factor of $\rank \ge r^{(\ref{lem:degree})}_d(c)$. Let $\Gamma:\T^c \rightarrow \T$ be an arbitrary function. Let $F:\F^n \rightarrow \T$ be defined by $F(x)=\Gamma(P_1(x),\ldots , P_c(x))$, and assume that $\deg(F)=d'$.
Then, for every collection of polynomials $Q_1,\ldots , Q_c:\F^n \rightarrow \T$ with $\deg(Q_i) \leq \deg(P_i)$ and $\depth(Q_i) \leq \depth(P_i)$, if $G:\F^n \rightarrow \T$ is defined by $G(x)=\Gamma(Q_1(x),\ldots , Q_c(x))$, then $\deg(G) \leq d'$.
\end{lem}

\section{Weak Regularity}\label{sec:weakreg}

Let $X$ and $Y$ be finite sets. Recall that $P(Y):=\{f:Y \rightarrow \R_{\geq 0}  :  \sum_{y \in Y}f(y)=1\}$ is the probability simplex on $Y$. As mentioned before, we shall write randomized functions by mapping them to the simplex. Thus for $f,g:X \rightarrow P(Y)$ we have $$\Pr_{x}[f(x)=g(x)]:=\E_x\langle f(x),g(x)\rangle.$$

\begin{lem}\label{lem:pseudorandom}Let $g:X \rightarrow P(Y)$, $\eps>0$, and $F$ be a collection of functions $f:X \rightarrow P(Y)$. Then there exist $c \leq 1/\eps^2$ functions $h_1, h_2,\ldots , h_c \in F$ and a function $\Gamma:P(Y)^c \rightarrow P(Y)$ such that for all $f \in F$, $$\l|\Pr[g(x)=f(x)]-\Pr[\Gamma(h_1(x),h_2(x),\ldots , h_c(x))=f(x)]\r| \leq \eps.$$
\end{lem}

\begin{proof}We construct $\calH=\{h_1, \ldots , h_c\} \subseteq F$ such that, if $\B_{\calH}$ is the factor of $X$ induced by $\calH$, then for all $f \in F$
$$\l|\Pr[\E[g|\B_{\calH}]=f(x)]-\Pr[g(x)=f(x)]\r| \leq \eps.$$
We then set $\Gamma:P(Y)^c \to P(Y)$ so that $\Gamma(h_1(x),\ldots , h_c(x))=\E[g|\B_{\calH}]$. In the following we shorthand $g_{\calH} = \E[g|\B_{\calH}]$. We consider the following variant of the Frieze-Kannan weak regularity algorithm \cite{FK99}.
\begin{itemize}
\item Initialize $\calH=\emptyset$
\item While there exists $f \in F$ such that $|\Pr[g_{\calH}(x)=f(x)]-\Pr[g(x)=f(x)]|>\eps$ \begin{itemize} \item Update $\calH=\calH \cup \{f\}$ \end{itemize}
\end{itemize}
The lemma follows from the following claim, which shows that we update $\calH$ at most $1/\eps^2$ times. Let $\|g_{\calH}\|_2^2:=\E_{x}\|g_{\calH}(x)\|_2^2$.

\begin{clm}\label{clm:friezekannan}
Consider any stage in the algorithm, with $\calH$ being the set of functions at that stage, and $f \in F$ being the new function added to $\calH$. Then
\begin{itemize}
\item $0 \leq \|g_{\calH}\|^2 \leq 1$;
\item $\|g_{\calH \cup \{f\}}\|^2 \geq \|g_{\calH}\|^2 + \eps^2$.
\end{itemize}
\end{clm}
\begin{proof}  The first part of the claim is trivial as $g_{\calH}$ maps to $P(Y)$.
For the second part, observe that $\langle g_{\calH \cup \{f\}}-g_{\calH},g_{\calH}\rangle=0$ and thus $$\|g_{\calH \cup \{f\}}\|_2^2=\|g_{\calH}\|_2^2+\|g_{\calH \cup \{f\}}-g_{\calH}\|_2^2$$
We will show that $\|g_{\calH \cup \{f\}}-g_{\calH}\|_2^2 \geq \eps^2$.
We have
\begin{align*}
\eps & < |\Pr[g_{\calH}(x)=f(x)]-\Pr[g(x)=f(x)]|\\
&=\l|\E_x\langle f(x),g_{\calH}(x)\rangle -\E_x\langle f(x),g(x)\rangle\r|\\
&=\l|\E_x\langle f(x),g_{\calH}(x)\rangle-\E_x\langle f(x),g_{\calH \cup \{f\}}(x)\rangle\r| \qquad \text{(as $f$ is measurable with respect to $\B_{\calH \cup \{f\}}$)}\\
&=\l|\E_x\langle f(x),g_{\calH}(x)-g_{\calH \cup \{f\}}(x)\rangle\r|\\
&\leq \E_x \l|\langle f(x),g_{\calH}(x)-g_{\calH \cup \{f\}}(x)\rangle \r|.
\end{align*}

Now, as $f:X \rightarrow P(Y)$, for every $x \in X$, $\|f(x)\|_2 \leq 1$. Thus, by the Cauchy-Schwartz inequality, for every $x \in X$, we have
$$|\langle f(x),g_{\calH}(x)-g_{\calH \cup \{f\}}(x)\rangle| \leq \|f(x)\|_2\|g_{\calH \cup \{f\}}(x)-g_{\calH}(x)\|_2 \leq \|g_{\calH \cup \{f\}}(x)-g_{\calH}(x)\|_2 $$
Thus, by another application of the Cauchy-Schwartz inequality, we have $$\eps^2 \leq \E_x \l|\langle f(x),g_{\calH}(x)-g_{\calH \cup \{f\}}(x)\rangle \r|^2 \leq \|g_{\calH \cup \{f\}}-g_{\calH}\|_2^2.$$
\end{proof}
This finishes the proof of the lemma.
\end{proof}

The following corollary for deterministic functions $f:X \to Y$ allows to obtain one-sided deterministic estimates. This simplifies some of the arguments later on.

\begin{cor}\label{cor:pseudorandom} Let $g:X \rightarrow Y$, $\eps>0$, and $F$ be a collection of functions $f:X \rightarrow Y$. Then there exist $c \leq 1/\eps^2$ functions $h_1, h_2,\ldots , h_c \in F$ such that for every $f \in F$, there is a function $\Gamma_f:Y^c \rightarrow Y$ such that $$\Pr_x[\Gamma_f(h_1(x), \ldots , h_c(x))=f(x)] \geq \Pr_x[g(x)=f(x)] - \eps.$$
\end{cor}

\begin{proof}
Applying Lemma~\ref{lem:pseudorandom} to $F$ we may assume the existence of $h_1,\ldots,h_c:X \to Y$ and $\Gamma:Y^C \to P(Y)$ such that for any $f \in F$,
$$
\l| \Pr[f(x)=\Gamma(h_1(x),\ldots,h_c(x))] - \Pr[f(x)=g(x)] \r| \le \eps.
$$
Let $A_{y_1,\ldots,y_c} = \{x \in X: h_1(x)=y_1,\ldots,h_c(x)=y_c\}$
be an atom defined by $h_1,\ldots,h_c$. Given $f \in F$, define $\Gamma_f:Y^c \to Y$ by letting $\Gamma_f(y_1,\ldots,y_c)$ to be the most common value that $f$ attains on $A_{y_1,\ldots,y_c}$.
Then
\begin{align*}
& \Pr[f(x)=\Gamma_f(h_1(x),\ldots,h_c(x))] \\
&= \sum_{y_1,\ldots,y_c \in Y}  \Pr[x \in A_{y_1,\ldots,y_c}] \cdot \max_{y^* \in Y} \Pr[f(x)=y^* | x \in A_{y_1,\ldots,y_c}]\\
&\ge \sum_{y_1,\ldots,y_c \in Y} \Pr[x \in A_{y_1,\ldots,y_c}] \cdot \Pr[f(x)=\Gamma(y_1,\ldots,y_c) | x \in A_{y_1,\ldots,y_c}]\\
& = \Pr[f(x)=\Gamma(h_1(x),\ldots,h_c(x))] \ge \Pr[f(x)=g(x)] - \eps.
\end{align*}
\end{proof}

\section{Proof of Theorem~\ref{THM:main}}

Fix a prime field $\F=\F_p$. For $d \in \N$, we shorthand $\de(d) = \de_{\F}(d)$. We restate Theorem~\ref{THM:main}.

\restate{Theorem \ref{THM:main}}
{Let $\eps>0$ and $d,n \in \N$. Then, $$\ell_{\F}(d,n,\de(d)-\eps) \leq c_{p,d,\eps}.$$}

We prove Theorem~\ref{THM:main} in the remainder of this section. Let $g:\F^n \rightarrow \U_1$ be a received word where we identify $\F$ with $\U_1$. Apply Corollary~\ref{cor:pseudorandom} with $X=\F^n$, $Y=\U_1$, $F=\Poly_{\leq d}(\F^n \rightarrow \U_1)$ and approximation parameter $\eps/2$ to obtain $\calH=\{h_1, \ldots , h_c\} \subseteq F$, $c \leq 4/\eps^2$ such that, for every $f \in F$, there is a function $\Gamma_f:\U_1^c \rightarrow \U_1$ satisfying
$$\Pr[\Gamma_f(h_1(x),h_2(x),\ldots , h_c(x))=f(x)] \geq \Pr[g(x)=f(x)] - \eps/2.$$

Let $r_1, r_2:\N \rightarrow \N$ be two non decreasing functions to be specified later, and let $C_{r,d}^{(\ref{lem:reg})}$ be as given in Lemma~\ref{lem:reg}. We will require that for all $m \ge 1$, \begin{equation}\label{eq:r1r2}  r_1(m)\geq r_2(C_{r_2,d}^{(\ref{lem:reg})}(m+1))+C_{r_2,d}^{(\ref{lem:reg})}(m+1)+1.
\end{equation}

As a first step, we $r_1$-regularize $\calH$ by Lemma~\ref{lem:reg}. This gives an $r_1$-regular factor $\B'$ of degree at most $d$, defined by polynomials $h_1',\ldots , h_{c'}':\F^n \rightarrow \T$, such that $\B' \succeq_{sem} \B$, $c' \leq C_{r_1,d}^{(\ref{lem:reg})}(c)$ and $\rank(\B')  \geq  r_1(c')$. We denote $\calH'=\{h_1',\ldots , h_{c'}'\}$. Note that $\calH'$ can have nonclassical polynomials as a result of the regularization. Let $\depth(h_i')=k_i$ for $i \in [c']$. Let $G_f:\otimes_{i=1}^{c'} \U_{k_i+1} \rightarrow \U_1$ be defined such that
$$
\Gamma_f(h_1(x),\ldots , h_c(x))=G_f(h_1'(x),\ldots , h_{c'}'(x)).
$$
Then
\begin{equation}\label{eq:fool}
\Pr[G_f(h_1'(x),h_2'(x),\ldots , h_{c'}'(x))=f(x)] \geq \Pr[g(x)=f(x)] - \eps/2.
\end{equation}
Next, given any classical polynomial $f:\F^n \rightarrow \T$ of degree at most $d$, we will show that if $\Pr[f(x)\neq g(x)] \leq \de(d)-\eps$, then $f$ is measurable with respect to $\calH'$ and this would upper bound the number of such polynomials by $p^{||\B'||}=p^{\prod_{i \in [c']}p^{k_i+1}}$ and as $c'=c'(p,d,\eps)$ and $k_i \leq \l\lfloor \frac{d-1}{p-1}\r\rfloor  $ this is independent on $n$.

Fix such a classical polynomial $f$. Appealing again to Lemma~\ref{lem:reg}, we $r_2$-regularize $\B_f:=\B' \cup \{f\}$. We get an $r_2$-regular factor $\B'' \succeq_{syn} \B'$ defined by the collection $\calH''=\{h_1',\ldots , h_{c'}',h''_1,\ldots , h''_{c''}\}\subseteq \Poly_{\leq d}(\F^n \rightarrow \T)$. Note that it is a syntactic refinement of $\B'$ as by our choice of $r_1$, $$\rank(\B') \geq r_1(c')  \geq  r_2(C_{r_2,d}^{(\ref{lem:reg})}(c'+1))+C_{r_2,d}^{(\ref{lem:reg})}(c'+1)+1 \geq r_2(|\B''|)+|\B''|+1.$$
We will choose $r_2$ such that for all $m \ge 1$,
\begin{equation}\label{eq:r2atom}
r_2(m) = \max\l(r_d^{(\ref{lem:atomsize})}\l(\frac{\eps/4}{\l(p^{\lfloor\frac{d-1}{p-1}\rfloor+1}\r)^m}\r),r^{(\ref{lem:degree})}_d(m)\r).
\end{equation}
Let $\depth(h''_j)=l_j$ for $j \in [c'']$ and denote $S:=\otimes_{i=1}^{c'}\U_{k_i+1} \times \otimes_{j=1}^{c''} \U_{l_j+1}$. Since $f$ is measurable with respect to $\B''$, there exists $F:S \rightarrow \U_1$ such that
$$
f(x)=F(h_1'(x),\ldots , h_{c'}'(x), h''_1(x),\ldots , h''_{c''}(x)).
$$
We next show that we can have each polynomial in the factor have a disjoint set of inputs, and still obtain more or less the same approximation factor.

\begin{clm}\label{clm:tild}Let $x^i, y^j$, $i \in [c'], j \in [c'']$ be pairwise disjoint sets of $n \in \N$ variables each. Let $n' = n(c'+c'')$. Let $\tilde{f}:\F^{n'} \rightarrow \U_1$ and $\tilde{g}:\F^{n'} \rightarrow \U_1$ be defined as   $$\tilde{f}(x)=F(h_1'(x^1),\ldots , h_{c'}'(x^{c'}), h''_1(y^1), \ldots , h''_{c''}(y^{c''}))$$ and $$\tilde{g}(x)=G_f(h_1'(x^1), \ldots , h_{c'}(x^{c'})).$$
Then $\deg(\tilde{f}) \le d$ and
$$
\l|\Pr_{x \in \F^{n'}}[\tilde{f}(x)=\tilde{g}(x)] - \Pr_{x \in \F^n}[f(x)=G_f(h_1'(x),h_2'(x),\ldots , h_c'(x))]\r|\leq  \eps/4.
$$
\end{clm}

\begin{proof}
The bound $\deg(\tilde{f}) \le \deg(f) \le d$ follows from Lemma~\ref{lem:degree} since $r_2(|\calH''|) \ge r^{(\ref{lem:degree})}_d(|\calH''|)$. To establish the bound on $\Pr[\tilde{f}=\tilde{g}]$, for each $s \in S$ let
$$
p_1(s) = \Pr_{x \in \F^n}[(h'_1(x),\ldots,h'_{c'}(x),h''_1(x),\ldots,h''_{c''}(x))=s].
$$
Applying Lemma~\ref{lem:atomsize} and since our choice of $r_2$ satisfies $\rank(\calH'') \ge r_d^{(\ref{lem:atomsize})}(\eps/4|S|)$, we have that $p_1$ is nearly uniform over $S$,
$$
p_1(s) = \frac{1 \pm \eps/4}{|S|}.
$$
Similarly, let
$$
p_2(s) = \Pr_{x^1,\ldots,x^{c'}, y^1,\ldots,y^{c''} \in \F^n}[(h'_1(x^1),\ldots,h'_{c'}(x^{c'}),h''_1(y^1),\ldots,h''_{c''}(y^{c''}))=s].
$$
Note that the rank of the collection of polynomials $\{h'_1(x^1),\ldots,h'_{c'}(x^{c'}),h''_1(y^1),\ldots,h''_{c''}(y^{c''})\}$ defined over $\F^{n'}$ cannot be lower than that of $\calH''$. Applying Lemma~\ref{lem:atomsize} again gives
$$
p_2(s) = \frac{1 \pm \eps/4}{|S|}.
$$
For $s \in S$, let $s' \in \otimes_{i=1}^{c'}\U_{k_i+1}$ be the restriction of $s$ to first $c'$ coordinates, that is, $s'=(s_1,\ldots ,s_{c'})$. Thus
\begin{align*}
\Pr_{x \in \F^{n'}}[\tilde{f}(x)=\tilde{g}(x)] &= \sum_{s \in S} p_2(s) 1_{F(s)=G_f(s')} \\
&= \sum_{s \in S} p_1(s) 1_{F(s)=G_f(s')} \pm \eps/4 \\
&= \Pr_{x \in \F^n}[f(x)=G_f(h_1'(x),h_2'(x),\ldots , h_c'(x))] \pm \eps/4.
\end{align*}
\end{proof}

So, we obtain that
$$
\Pr_{x \in \F^{n'}}[\tilde{f}(x)=\tilde{g}(x)] \ge \Pr_{x \in \F^n} [f(x) = G_f(h'_1(x),\ldots,h'_{c'}(x))] - \eps/4 \ge 1 - \de(d)+\eps/4.
$$
Next, we need the following variant of the Schwartz-Zippel lemma \cite{Sch, Zip}.
\begin{clm}\label{clm:sz1}
Let $d,n_1,n_2 \in \N$. Let $f_1:\F^{n_1+n_2} \rightarrow \F$ and $f_2:\F^{n_1} \rightarrow \F$ be such that $\deg(f_1)\leq d$ and $$\Pr[f_1(x_1,\ldots , x_{n_1+n_2})=f_2(x_1,\ldots , x_{n_1})]>1-\de(d)$$
Then, $f_1$ does not depend on $x_{n_1+1}, \ldots , x_{n_1+n_2}$.
\end{clm}
\begin{proof}
We will show that $f_1$ does not depend on $z=x_{n_1+n_2}$ say. The proof for any other variable is similar. Recall that $\de(d):=\frac{1}{p^a}\l(1-\frac{b}{p}\r)$ where $d=a\cdot (p-1)+b$.
Let $f_1(x)=\sum_{k=0}^{d'}c_kz^k$ where $c_k \in \F[x_1,\ldots , x_{n_1+n_2-1}]$ and $d' \leq \min \{d,p-1\}$. Then $(f_1-f_2)(x)=c_0-f_2(x)+\sum_{k=1}^{d'}c_kz^k$. We will show that $d' \geq 1$ will lead to a contradiction. Let $\deg(c_{d'})=d''$. Note that $d''+d' \leq d$.  Then, $$\Pr[(f_1-f_2)(x)=0]\leq \Pr[c_{d'}=0]+(1-\Pr[c_{d'}=0])(1-\de(d')) \leq 1-\de(d'')\de(d').$$
We will show that for any $d \ge 1$ and any $1 \le c \leq p-1$, we have $\de(c)\de(d-c) \geq \de(d)$ and this will show that $\Pr[(f_1-f_2)(x)=0] \leq 1-\de(d'+d'')\leq 1-\de(d)$ which leads to a contradiction. Thus, $f_1$ will not depend on $z$.
We will now show that \begin{equation}\label{eq:sz}\de(c)\de(d-c) \geq \de(d)\end{equation} Let $d=a\cdot (p-1)+b$.
\paragraph{Case 1: $0 \leq c \leq b$}
\begin{eqnarray*}
\eqref{eq:sz}&\Leftrightarrow &\l(1-\frac{c}{p}\r)\frac{1}{p^a}\l(1-\frac{b-c}{p}\r)\geq \frac{1}{p^a}\l(1-\frac{b}{p}\r)\\
&\Leftrightarrow & b \geq c
\end{eqnarray*}

\paragraph{Case 2: $b < c \leq p-1$}
\begin{eqnarray*}
\eqref{eq:sz}&\Leftrightarrow &\l(1-\frac{c}{p}\r)\frac{1}{p^{a-1}}\l(\frac{1+c-b}{p}\r)\geq \frac{1}{p^a}\l(1-\frac{b}{p}\r)\\
&\Leftrightarrow & (c-b)\l(1-\frac{c+1}{p}\r)\geq 0
\end{eqnarray*}
which is true by hypothesis.
\end{proof}

Now apply Claim \ref{clm:sz1} to $f_1=\tilde{f}, f_2=\tilde{g}, n_1=n c', n_2= nc''$. We obtain that $\tilde{f}$ does not depend on $y^1,\ldots,y^{c''}$. Hence,
$$
\tilde{f}(x^1,\ldots,x^{c'},y^1,\ldots,y^{c''})=F(h_1'(x^1),\ldots , h_{c'}'(x^{c'}), C_1 ,\ldots , C_{c''})$$
where $C_j=h''_j(0) \in \U_{l_j+1}$ for $j \in [c'']$. If we substitute $x^1=\ldots=x^{c'}=x$ we get that
$$
f(x)=F(h'_1(x),\ldots,h'_{c'}(x),h''_1(x),\ldots,h''_{c''}(x)) =
F(h_1'(x),\ldots , h_{c'}'(x), C_1, \ldots , C_{c''}),
$$
which shows that $f$ is measurable with respect to $\calH'$, as claimed.

\section{Proof of Theorem~\ref{THM:main2}}

\restate{Theorem~\ref{THM:main2}}
{Let $\F=\F_p$ be a prime field. Let $\eps>0$ and $e \leq d,n \in \N$. Then, $$\ell_{\F}(d,n,\de(e)-\eps) \leq \exp\l(c_{p,d,\eps}n^{d-e}\r).$$}

The proof follows along the same lines as that of Theorem~\ref{THM:main}. It will rely on the following lemma which generalizes Claim~\ref{clm:sz1}.

\begin{lem}\label{lem:main}
Fix $d \ge e \ge 1, \eps>0$. There exists $r^{(\ref{lem:main})}_{d,\eps} \in \N$ such that the following holds. Let $f_1:\F^{n_1+n_2} \rightarrow \U_1$ be a classical polynomial of degree at most $d$. Assume that
\begin{itemize}
\item There exist $f_2:\F^{n_1}\rightarrow \U_1$ be such that $\Pr[f_1(x,y)=f_2(x)] \geq 1-\de(e)+\eps$.
\item There exists a polynomial $h:\F^{n_2} \to \U_{k+1}$ of degree at most $d$ such that the factor it defines has rank at least $r^{(\ref{lem:main})}_{d,\eps}$, and a function $\Gamma:\F^{n_1} \times \U_{k+1} \to \U_1$, such that
$$
f_1(x,y) = \Gamma(x,h(y)).
$$
\item The dependence on the depth of $h$ is nontrivial: $f_1(x,y)$ cannot be written as $\Gamma'(x,p \cdot h(y))$ for any $\Gamma':\F^{n_1} \times \U_{k} \to \U_1$.
\end{itemize}
Then $\deg(h) \leq d-e$.
\end{lem}

We first prove Theorem~\ref{THM:main2} assuming Lemma~\ref{lem:main}.

\begin{proof}[Proof of Theorem~\ref{THM:main2} assuming Lemma~\ref{lem:main}]
The initial part of the proof is as in Theorem~\ref{THM:main}. Assume that $n>r^{(\ref{lem:main})}_{d,\eps/4}$ otherwise the theorem is trivially true. Let $f,g:\F^n \rightarrow \U_1$ with $\deg(f) \leq d$ and $\dist(f,g) \leq \de(e)-\eps$. For non decreasing functions $r_1, r_2:\N \rightarrow \N$, chosen as in the proof of Theorem~\ref{THM:main}, we have an $r_1$-regular $\calH'=\{h_1',\ldots , h_{c'}'\}$ and an $r_2$-regular $\calH''=\calH' \cup \{h''_1,\ldots , h''_{c''}\}$ where each $h_i',h''_i$ is a nonclassical polynomial of degree $\leq d$, such that the following holds.

Let $\depth(h_i')=k_i$ for $i \in [c']$ and $\depth(h''_j)=l_j$ for $j \in [c'']$. Since $f$ is measurable with respect to $\calH''$, there exists $F:\otimes_{i=1}^{c'}\U_{k_i+1} \times \otimes_{j=1}^{c''} \U_{l_j+1}  \rightarrow \U_1$ such that
$$
f(x)=F(h_1'(x),\ldots , h_{c'}'(x), h''_1(x),\ldots , h''_{c''}(x)).
$$
We may assume that for all $i \in [c'']$, the depth of $h''_i$ is minimal, in the sense that we cannot replace $h''_i$ with $p \cdot h''_i$ and change $F$ accordingly to still compute $f$ (if this is not the case, then replace $h''_i$ with $p \cdot h''_i$ whenever possible; this only reduces the degree of $h''_i$ and the new factor has rank at least that of the original factor). Also, there exists a function $G_f:\otimes_{i=1}^{c'}\U_{k_i+1}  \rightarrow \U_1$ such that
$$
\Pr[G_f(h'_1(x),\ldots,h'_{c'}(x))=f(x)] \ge 1-\de(e)+\eps/2.
$$
We will show that this implies that $\deg(h''_i) \le d-e$ for all $i \in [c'']$. Let $\B''$ be the factor defined by $\calH''$. As the number of polynomials of degree $d-e$ is exponential in $n^{d-e}$, the number of functions $f$ is controlled by the product of the number of composing functions $F$, which is $p^{||\B''||}=p^{\l(\prod_{i \in [c']}p^{k_i+1}\r)\l(\prod_{j \in [c'']}p^{l_j+1}\r)}=c_1(p,d,\eps)$, and the number of choices for $h''_1,\ldots,h''_{c''}$, which is $\exp(c_2 c'' n^{d-e})$. This amounts to at most $\exp(c n^{d-e})$ for some $c=c(p,d,\eps)$, as claimed.

To prove the bound on the degrees of $h''_1,\ldots,h''_{c''}$, define, as in the proof of Theorem~\ref{THM:main}, $x^i, y^j$ for $i \in [c'], j \in [c'']$ to be pairwise disjoint sets of $n \in \N$ variables. Let $n' = n(c'+c'')$. Define $\tilde{f}:\F^{n'} \rightarrow \U_1$ and $\tilde{g}:\F^{n'} \rightarrow \U_1$ as $$\tilde{f}(x^1,\ldots,x^{c'},y^1,\ldots,y^{c''})=F(h_1'(x^1),\ldots , h_{c'}'(x^{c'}), h''_1(y^1), \ldots , h''_{c''}(y^{c''}))$$ and $$\tilde{g}(x^1,\ldots,x^{c'})=G_f(h_1'(x^1), \ldots , h_{c'}(x^{c'})).$$
Then, by Claim~\ref{clm:tild}, $\deg(\tilde{f}) \leq d$ and $\Pr[\tilde{f}=\tilde{g}] \ge 1-\de(e)+\eps/4$.

We next apply Lemma~\ref{lem:main} to show that $\deg(h''_j) \le d-e$ for all $j \in [c'']$. To see that for say, $j=c''$, let $k=\depth(h''_{c''})$, $n_1=n(c'+c''-1), n_2=n, h(y)=h''_{c''}(y)$ and $\Gamma:\F^{n_1} \times \U_{k+1} \to \U_1$ given by
$$
\Gamma((x^1,\ldots,x^{c'},y^1,\ldots,y^{c''-1}),\alpha) = F(h_1'(x^1), \ldots , h'_{c'}(x^{c'}),h''_1(y^1),\ldots,h''_{c''-1}(y^{c''-1}), \alpha).
$$
so that
$$
\tilde{f}(x^1,\ldots,x^{c'},y^1,\ldots,y^{c''}) = \Gamma((x^1,\ldots,x^{c'},y^1,\ldots,y^{c''-1}),h''_{c''}(y^{c''})).
$$
If we make sure that $r_2(m) \ge r^{(\ref{lem:main})}_{d,\eps/4}$ for all $m \ge 1$, then we establish all the requirements for Lemma~\ref{lem:main}. Hence we deduce that $\deg(h''_{c''}) \le d-e$ as claimed.
\end{proof}

\subsection{Proof of Lemma~\ref{lem:main}}

We prove Lemma~\ref{lem:main} in this section. Fix $d \ge e \ge 1$ and $\eps>0$. Let $r=r^{(\ref{lem:main})}_{d,\eps}$ be large enough to be chosen later. We first show that we can replace $h$ with a simple polynomial of the same degree and depth, which would allow us to simplify the analysis.

Let $\depth(h)=k$ and let $A=\deg(h)-(p-1)k$. Define $\tilde{h}:\F^{rA}\rightarrow \U_{k+1}$ as follows. Let $z=(z_{1,1}, \ldots , z_{r,A}) \in \F^{rA}$ and define
\begin{equation}
\tilde{h}(z):=\frac{\sum_{i=1}^r \prod_{j=1}^A  z_{i,j}}{p^{k+1}}.
\end{equation}
Note that $\tilde{h}$ and $h$ are both polynomials of the same degree and depth. Define $\tilde{f}_1:\F^{n_1+rA}\rightarrow \U_1$ as
$$
\tilde{f}_1(x,z)=\Gamma(x,\tilde{h}(z)).
$$

We will show that we may analyze $\tilde{f}_1$ instead of $f_1$ to obtain the upper bound on $\deg(h)$. To simplify the presentation,
denote $Z_i := \prod_{j=1}^A z_{i,j}$ for $i \in [r]$. First, we argue that if $r$ is chosen large enough then both $h,\tilde{h}$ are nearly uniform over $\U_{k+1}$.

\begin{clm}\label{clm:h_htilde_uniform}
If $r$ is chosen large enough then for all $\alpha \in \U_{k+1}$,
$$
\Pr_{y \in \F^{n_2}}[h(y)=\alpha] = p^{-(k+1)} (1 \pm \eps/2)
$$
and
$$
\Pr_{z \in \F^{rA}}[\tilde{h}(z)=\alpha] = p^{-(k+1)} (1 \pm \eps/2).
$$
\end{clm}

\begin{proof}
The proof for $h$ follows from Lemma~\ref{lem:atomsize} by choosing
$r \geq r_d^{\ref{lem:atomsize}}\l(\frac{\eps}{2p^{k+1}}\r)$. The proof for $\tilde{h}$ follows by a simple Fourier calculation. Let $\omega=\exp(2 \pi i / p^{k+1})$. We have $\Pr[Z_i=0],\Pr[Z_i=1] \ge p^{-A}\ge p^{-d}$.   One can verify that this implies that for any nonzero $c \in \Z_{p^{k+1}}$, $\E[\omega^{c Z_i}] \le 1-\eta$ for $\eta = p^{-O(d)}$. As $Z_1,\ldots,Z_r$ are independent we have
$\E\l[\omega^{c(Z_1+\ldots+Z_r)}\r] \le (1-\eta)^r$. Hence if we choose $r$ large enough so that $(1-\eta)^r < (\eps/2) p^{-(k+1)}$ then, for any $a \in \Z_{p^{k+1}}$,
\begin{align*}
\Pr[Z_1+\ldots+Z_r=a \pmod{p^{k+1}}] &= p^{-(k+1)}\l(1 + \sum_{c \in \Z_{p^{k+1}} \setminus \{0\}} \omega^{-ac} \cdot \E\l[\omega^{c(Z_1+\ldots+Z_r)}\r]\r) \\&= p^{-(k+1)} (1 \pm \eps/2).
\end{align*}
\end{proof}

This implies that $f_2(x)$ is also well approximates $\tilde{f}_1(x,z)$.

\begin{cor}\label{cor:f2_approx_f1tilde}
$\Pr[\tilde{f}_1(x,z)=f_2(x)] \ge \Pr[f_1(x,y) = f_2(x)] - \eps/2 \ge 1-\de(e)+\eps/2$ where $x \in \F^{n_1}, y \in \F^{n_2},  z \in \F^{rA}$ are chosen uniformly and independently.
\end{cor}

\begin{proof}
Claim~\ref{clm:h_htilde_uniform} implies that the statistical distance between $h(y)$ and $\tilde{h}(z)$ is at most $\eps/2$. Hence for every fixed $x$, $|\Pr[\Gamma(x,h(y))=f_2(x)] - \Pr[\Gamma(x,\tilde{h}(z))=f_2(x)]| \le \eps/2$.
\end{proof}

We next argue that by choosing $r$ large enough, we can guarantee that $\tilde{f}_1$ has degree at most $d$.

\begin{clm}\label{clm:degreethm2}If $r$ is chosen large enough then $\deg(\tilde{f}_1)\leq \deg(f_1) \leq d$.
\end{clm}
\begin{proof} By Claim~\ref{clm:h_htilde_uniform}, if $r$ is chosen large enough then $h(y),\tilde{h}(z)$ attain all possible values in $\U_{k+1}$. For every $\alpha  \in \U_{k+1}$, let $f_{\alpha}(x):=\Gamma(x,\alpha)$. Note that as there exists some $y_{\alpha} \in h^{-1}(\alpha)$ then $f_{\alpha}(x)=f_1(x,y_{\alpha})$ is a (classical) polynomial in $x$ of degree at most $d$.

We have $f_1(x,y)=\Gamma(x,h(y))=\Gamma'((f_{\alpha}(x):\alpha \in \U_{k+1}),h(y))$ for some $\Gamma':\F^{p^{k+1}} \times \U_{k+1} \to \F$. Let $\calH=\{f_{\alpha}(x):\alpha \in \U_{k+1}\}$ and for $r_1:\N \rightarrow \N$ a growth function to be specified later, let $\calH'=\{g_1(x),\ldots,g_c(x)\}$ be the result of $r_1$-regularizing $\calH$ by Lemma~\ref{lem:reg}. Then
$$
f_1(x,y) = \Gamma''(g_1(x),\ldots,g_c(x),h(y))
$$
for some $\Gamma'':\F^c \times \U_{k+1} \to \F$. Hence also
$$
\tilde{f}_1(x,z) = \Gamma(x,\tilde{h}(z)) = \Gamma''(g_1(x),\ldots,g_c(x),\tilde{h}(z)).
$$
We next apply Lemma~\ref{lem:degree} to bound the degree of $\tilde{f}_1$. This requires to assume that $r_1(c) \ge r_d^{(\ref{lem:degree})}(c+1)$ and $r \ge r_d^{(\ref{lem:degree})}(C_{r_1,d}^{(\ref{lem:reg})}(p^{k+1})+1)$. We obtain that
$$
\deg(\tilde{f}_1) = \deg(\Gamma''(g_1(x),\ldots,g_{c}(x),\tilde{h}(z))) \le\deg(\Gamma''(g_1(x),\ldots,g_{c}(x),h(y)))=\deg(f_1)=d.
$$
\end{proof}

We next analyze the specific properties of $\tilde{h}$. Recall that we set
$Z_i:=\prod_{j=1}^A z_{i,j}$ so that $\tilde{h}(z) = \frac{\sum Z_i}{p^{k+1}}$. Since $\tilde{h}$ depends only on $W=\sum Z_i \mod p^{k+1}$, let the digits of $W \mod p^{k+1}$ in base $p$, be represented by classical polynomials $W_0(z),\ldots , W_k(z):\F^{rA} \to \F$. Then, we can express $\tilde{f_1}(x,z)$ as
\begin{equation}
\tilde{f_1}(x,z) = \Gamma(x, \tilde{h}(z)) = \Gamma'(x,W_0(z),W_1(z),\ldots,W_{k}(z))
\end{equation}
for some $\Gamma':\F^{n_1} \times \F^{k+1} \to \U_1$. Recall that we assumed that $\Gamma$ depends nontrivially on the depth of
its second argument. This implies that $\Gamma'$ depends nontrivially on its last
input (i.e. $W_k(z))$.  As $\tilde{f_1}$ is a classical polynomial, and each $W_{i}$ take values in $\F$,
identifying $\U_1$ with $\F$, we can decompose
\begin{equation}\label{eq:symexpansion}
\tilde{f_1}(x,z) = \sum_{0 \le d_0,\ldots,d_k \le p-1} f_{d_0,\ldots,d_k}(x) \prod_{i=0}^k W_{i}(z)^{d_i},
\end{equation}
where $f_{d_0,\ldots,d_k} \in \F[x]$ is a classical polynomial. We next argue that $\deg(f_{d_0,\ldots,d_k})$ cannot be too large.

\begin{lem}\label{lem:deg_coef}
$\deg(f_{d_0,\ldots,d_k}) \le d - A \sum_{i=0}^k p^i d_i$ for all $0 \le d_0,\ldots,d_k \le p-1$.
\end{lem}

We will require a few simple claims first. The $\ell$-th symmetric polynomial in $Z=(Z_1,\ldots,Z_r)$, for $1 \le \ell \le r$, is a classical polynomial of degree $\ell$ defined as
$$
S_{\ell}(Z)=\sum_{1 \leq i_1 < \ldots < i_{\ell} \leq r}\prod_{j=1}^{\ell} Z_{i_j}.
$$

For $0 \le i \le k$, define $W_i':\F^{rA}\to \F$ by $W_i'(z):=S_{p^i}(Z)$. The following claim follows immediately from Lucas theorem \cite{Lucas}.
\begin{clm}\label{clm:lucas}
Let $z \in \{0,1\}^{rA}$. Then, $W_i(z)=W_i'(z)$ for $i=0,\ldots,k$.
\end{clm}

\begin{proof}
If $z \in \{0,1\}^{rA}$ then $Z \in \{0,1\}^r$. Lucas theorem implies that the $i$-th least significant digit (starting at $0$) of $W=Z_1+\ldots+Z_r$ in base $p$ is given by ${Z_1+\ldots+Z_r \choose p^i} \mod p = S_{p^i}(Z)$.
\end{proof}

For every polynomial $P \in \F[z]$, define $\ML(P)$ to be the multilinearization of $P$. That is, it is obtained by replacing each $z_{i,j}^a$ by $z_{i,j}$ for all $a \ge 1$ and all $i \in [r], j \in [A]$. Note that $\ML(P)(z)=P(z)$ for all $z \in \{0,1\}^{rA}$.

\begin{clm}\label{clm:multilinear}  Let $P,Q:\F^{rA} \rightarrow \F$ be two polynomials such that $P(z)=Q(z)$ for all $z \in \{0,1\}^{rA}$. Then $\ML(P)  \equiv \ML(Q)$. 
\end{clm}

\begin{proof}
Let $n=rA$. It is easy to see that a multilinear polynomial $f: \F^n \rightarrow \F$ satisfies $f(z)=0$ for all $z \in \{0,1\}^n$ if and only if $f \equiv 0$. Therefore, for every polynomial $P:\F^n \rightarrow \F$, $\ML(P)$ is the unique multilinear polynomial that agrees with $P$ on $\{0,1\}^n$.  Let $R:\F^n \rightarrow \F$ be defined as $R:=P-Q$. Then by linearity, $\ML(R):\equiv \ML(P)-\ML(Q)$. As $\ML(R)=0$ for all $z \in \{0,1\}^n$, $\ML(R) \equiv 0$ which implies $\ML(P)\equiv \ML(Q)$.
\end{proof}

\begin{proof}[Proof of Lemma~\ref{lem:deg_coef}]
For $D=\sum_{i=0}^k p^i d_i$, define $$W^{(D)}(z):=\prod_{i=0}^k W_{i}(z)^{d_i}, \qquad W'^{(D)}(z):=\prod_{i=0}^k W_{i}'(z)^{d_i}.$$
By Claim~\ref{clm:lucas} and Claim~\ref{clm:multilinear}, we can define a common multilinearization of $W^{(D)}$ and $W'^{(D)}$ by
$$ M^{(D)}:=\ML\l(W^{(D)}\r)=\ML\l(W'^{(D)}\r).$$
Let $m'(z)=\prod_{i=1}^D Z_i
= \prod_{i=1}^D \prod_{j=1}^A z_{i,j}$ be a monomial. The coefficient of $m'$ in $W'^{(D)}$ is equal to the coefficient of  $\prod_{i=1}^D Z_i$ in $\prod_{i=0}^k S_{p^i}(Z)^{d_i}$, which is equal to the number of partitions of a set of size $D$ to $d_0$ sets of size $1$, $d_1$ sets of size $p$, $d_2$ sets of size $p^2$, up to $d_k$ sets of size $p^k$. This is given by
$$
\prod_{i=0}^k \prod_{j=1}^{d_i} {j p^i + d_{i+1} p^{i+1} + \ldots + d_k p^k \choose p^i},
$$
which by Lucas theorem is equal modulo $p$ to $\prod_{i=0}^k (d_i !) \ne 0 \mod p$.

Owing to the above, we have $\deg(M^{(D)})\leq \deg(W'^{(D)})=AD$. Also, since $m'(z)$ is of maximal degree, it also remains in $M^{(D)}$ after multilinearization.
Define $$\bar{f_1}(x,z):=\sum_{0 \le d_0,\ldots,d_k \le p-1} f_{d_0,\ldots,d_k}(x)M^{(D)}(z).$$
Then, we have $\deg(\bar{f_1}) \leq \deg(\tilde{f}_1) \leq d$.

Now, suppose that the lemma is false. Let $D=\sum p^i d_i$ be maximal such that $\deg(f_{d_0,\ldots,d_k}) > d - A D$. Note that $D$ corresponds to a unique tuple $(d_0,\ldots,d_k)$. Let $m(x)$ be any monomial in $f_{d_0,\ldots,d_k}(x)$ with maximal degree, and recall that $m'(z)=\prod_{i=1}^D Z_i
= \prod_{i=1}^D \prod_{j=1}^A z_{i,j}$. Hence, the monomial $m(x) m'(z)$, whose degree is larger than $d$, has a nonzero coefficient in $f_{d_0,\ldots,d_k}(x) M^{(D)}(z)$ as noted above. We will show it has a zero coefficient in any other $f_{d'_0,\ldots,d'_k}(x) M^{(D')}(z)$ with $(d'_0,\ldots,d'_k) \ne (d_0,\ldots,d_k)$, $D'=\sum_{i}p^i d_i'$ which will contradict the fact that $\deg(\bar{f}_1) \le d$.

So, let $(d'_0,\ldots,d'_k) \ne (d_0,\ldots,d_k)$ and let $D' = \sum p^i d'_i$. Note that necessarily $D' \ne D$. If $D' > D$ then by maximality of $D$, $\deg\l(f_{d'_0,\ldots,d'_k}\r) \le d-AD' < d-AD$ and hence $m(x)$ cannot appear in $f_{d'_0,\ldots,d'_k}(x)$. If $D' < D$ then $\deg\l(M^{(D')}\r) = AD' < AD$ and hence $m'(z)$ cannot appear in $M^{(D')}(z)$.
\end{proof}

Let $w=(w_0,\ldots,w_k) \in \F^{k+1}$ be new variables, and define $f'_1:\F^{n_1+k+1} \to \F$ by
\begin{equation}\label{eq:f1prime_def}
f'_1(x,w) = \Gamma'(x,w_0,\ldots,w_k) = \sum_{0 \le d_0,\ldots,d_k \le p-1} f_{d_0,\ldots,d_k}(x) \prod_{i=0}^k w_i^{d_i}.
\end{equation}

We next argue that $f'_1$ is also well approximated by $f_2$.

\begin{clm}
$\Pr[f'_1(x,w)=f_2(x)] \ge \Pr[\tilde{f}_1(x,z) = f_2(x)] - \eps/4 \ge 1-\de(e)+\eps/4$, where $x \in \F^{n_1}, z \in \F^{rA}, w \in \F^{k+1}$ are uniformly and independently distributed.
\end{clm}

\begin{proof}
By Claim~\ref{clm:h_htilde_uniform}, the distribution of $\tilde{h}$ is $\eps/4$-close close in statistical distance to the uniform distribution over $\U_{k+1}$, hence the distribution of $(W_0(z),\ldots,W_k(z))$ is $\eps/4$-close in statistical distance to the uniform distribution over $\F^{k+1}$.
\end{proof}

To conclude the proof of Lemma~\ref{lem:main}, expand $f'_1-f_2$ as
$$
f'_1(x,w) - f_2(x) = \sum_{i=0}^{d'} c_i(x,w_0,\ldots,w_{k-1}) w_k^i
$$
where $c_i \in \F[x,w_0,\ldots,w_{k-1}]$, $d' \le \min(d,p-1)$ and $c_{d'} \ne 0$. We have that $d' \ge 1$ since $\Gamma'$ depends on $W_k(z)$. Also, by Lemma~\ref{lem:deg_coef}, for $i \ge 1$ we have $\deg(c_i) \le d - A p^k i$. To see this, suppose not. Consider the expansion in \eqref{eq:f1prime_def}. Then, for some $d_0,\ldots d_{k-1}$, $\deg(f_{d_0,\ldots d_{k-1},i})+\sum_{j=0}^{k-1}d_j > d-Ap^ki$, which implies that
$$
\deg(f_{d_0,\ldots d_{k-1},i}) >   d- \sum_{j=0}^{k-1}d_j - A p^k i \ge d - A \sum_{j=0}^{k-1} d_j p^j - A p^k i,
$$
which is a contradiction to Lemma~\ref{lem:deg_coef}. Hence
\begin{align*}
\Pr[f'_1(x,w)=f_2(x)] &\le \Pr[c_{d'}=0] + (1-\Pr[c_{d'}=0]) (1-\de(d')) \\
&\le 1 - \de(d-A p^k d') \de(d') \le 1 - \de(d-d' (A p^{k}-1)),
\end{align*}
where the last inequality was established in Claim~\ref{clm:sz1}. So, as we established that $\de(d-d' (A p^k-1)) < \de(e)$ and $d' \ge 1$ we must have $A p^k - 1 < d-e$, and hence $A p^k \le d-e$. Now, recall that $\deg(h)=\deg(\tilde{h})=A+(p-1)k$ and it is a simple exercise to verify that $A+(p-1)k \le A p^k$ for all $A \ge 1, k \ge 0$. We thus showed that $\deg(h) \le d-e$, as claimed.

\section{Open Problems}
Theorem~\ref{THM:main} and Theorem~\ref{THM:main2} establish that over any fixed prime field $\F_p$ and any fixed $e \le d$ and $\eps>0$, the number of degree $d$ polynomials in a any ball of radius $\de(e)-\eps$ is at most $\exp(c n^{d-e})$ for some $c=c(p,d,\eps)$, which in particular resolves the conjecture raised in \cite{GKZ08} when $e=d$.

However, the bounds on $c$ which we obtain are of Ackermann-type, which seem far from optimal. This leaves open the question of obtaining better bounds. This may require a different approach, as currently higher-order Fourier analysis does not seem to provide better bounds. We also leave as an open problem the question of extending our work to non-prime fields, and note that the missing ingredient is an extension of the higher-order Fourier analytic techniques to non prime fields.

\bibliographystyle{alpha}
\bibliography{listdecoding}

\newcommand{\etalchar}[1]{$^{#1}$}
\begin{thebibliography}{BFH{\etalchar{+}}13}

\bibitem[AGS03]{AGS}
A.~Akavia, S.~Goldwasser, and S.~Safra.
\newblock Proving hard-core predicates using list decoding.
\newblock In {\em Proc. $44^{th}$ IEEE Symposium on Foundations of Computer
  Science (FOCS'03)}, 2003.

\bibitem[AS03]{AroraSudan}
S.~Arora and M.~Sudan.
\newblock Improved low-degree testing and its applications.
\newblock {\em Combinatorica}, 23(3):365--426, 2003.

\bibitem[BFH{\etalchar{+}}13]{BFHHL}
Arnab Bhattacharyya, Eldar Fischer, Hamed Hatami, Pooya Hatami, and Shachar
  Lovett.
\newblock Every locally characterized affine-invariant property is testable.
\newblock In {\em STOC}, pages 429--436, 2013.

\bibitem[DGKS08]{DGKS08}
Irit Dinur, Elena Grigorescu, Swastik Kopparty, and Madhu Sudan.
\newblock Decodability of group homomorphisms beyond the johnson bound.
\newblock In {\em STOC}, pages 275--284, 2008.

\bibitem[Eli57]{Elias}
P.~Elias.
\newblock List decoding for noisy channels.
\newblock Technical Report 335, Research Laboratory of Electronics, MIT, 1957.

\bibitem[FK99]{FK99}
Alan~M. Frieze and Ravi Kannan.
\newblock Quick approximation to matrices and applications.
\newblock {\em Combinatorica}, 19(2):175--220, 1999.

\bibitem[GKZ08]{GKZ08}
P.~Gopalan, A.~Klivans, and D.~Zuckerman.
\newblock List decoding {R}eed-{M}uller codes over small fields.
\newblock In {\em Proc. $40^{th}$ ACM Symposium on the Theory of Computing
  (STOC'08)}, pages 265--274, 2008.

\bibitem[GL89]{GL}
O.~Goldreich and L.~Levin.
\newblock A hard-core predicate for all one-way functions.
\newblock In {\em Proc. $21^{st}$ ACM Symposium on the Theory of Computing},
  pages 25--32, 1989.

\bibitem[Gop10]{Gopalan10}
P.~Gopalan.
\newblock A {F}ourier-analytic approach to {R}eed-{M}uller decoding.
\newblock In {\em Proc. $51^{st}$ IEEE Symp. on Foundations of Computer Science
  (FOCS'10)}, pages 685--694, 2010.

\bibitem[GRS00]{GRS}
O.~Goldreich, R.~Rubinfeld, and M.~Sudan.
\newblock Learning polynomials with queries: The highly noisy case.
\newblock {\em SIAM J. Discrete Math.}, 13(4):535--570, 2000.

\bibitem[GS99]{GuruswamiSudan}
V.~Guruswami and M.~Sudan.
\newblock Improved decoding of {R}eed-{S}olomon and {A}lgebraic-{G}eometric
  codes.
\newblock {\em IEEE Transactions on Information Theory}, 45(6):1757--1767,
  1999.

\bibitem[GT09]{GT09}
B.~Green and T.~Tao.
\newblock The distribution of polynomials over finite fields, with applications
  to the gowers norms.
\newblock {\em Contrib. Discrete Math}, 4(2):1--36, 2009.

\bibitem[Gur04]{Venkat:thesis}
V.~Guruswami.
\newblock {\em List Decoding of Error-Correcting Codes}, volume 3282 of {\em
  Lecture Notes in Computer Science}.
\newblock Springer, 2004.

\bibitem[Gur06]{Venkat:book}
V.~Guruswami.
\newblock {\em Algorithmic Results in List Decoding}, volume~2 of {\em
  Foundations and Trends in Theoretical Computer Science}.
\newblock Now Publishers, 2006.

\bibitem[Jac97]{Jackson}
J.~Jackson.
\newblock An efficient membership-query algorithm for learning {DNF} with
  respect to the uniform distribution.
\newblock {\em Journal of Computer and System Sciences}, 55:414--440, 1997.

\bibitem[KLP10]{KLP10}
T.~Kaufman, S.~Lovett, and E.~Porat.
\newblock Weight distribution and list-decoding size of {R}eed-{M}uller codes.
\newblock In {\em Innovations in Computer Science (ICS'10)}, pages 422--433,
  2010.

\bibitem[KM93]{KM}
E.~Kushilevitz and Y.~Mansour.
\newblock Learning decision trees using the {F}ourier spectrum.
\newblock {\em SIAM Journal of Computing}, 22(6):1331--1348, 1993.

\bibitem[Luc78]{Lucas}
Edouard Lucas.
\newblock Théorie des fonctions numériques simplement périodiques.
\newblock {\em American Journal of Mathematics}, 1(2):pp. 184--196, 1878.

\bibitem[PW04]{PellikaanWu}
R.~Pellikaan and X.~Wu.
\newblock List decoding of q-ary {R}eed-{M}uller codes.
\newblock {\em IEEE Transactions on Information Theory}, 50(4):679--682, 2004.

\bibitem[Sch80]{Sch}
J.~T. Schwartz.
\newblock Fast probabilistic algorithms for veriﬁcation of polynomial
  identities.
\newblock {\em Journal of the ACM}, 27:701--717, 1980.

\bibitem[STV01]{STV}
M.~Sudan, L.~Trevisan, and S.~P. Vadhan.
\newblock Pseudorandom generators without the {XOR} lemma.
\newblock {\em J. Comput. Syst. Sci.}, 62(2):236--266, 2001.

\bibitem[SU05]{SU}
Ronen Shaltiel and Christopher Umans.
\newblock Simple extractors for all min-entropies and a new pseudorandom
  generator.
\newblock {\em J. ACM}, 52(2):172--216, 2005.

\bibitem[Sud97]{Sudan}
M.~Sudan.
\newblock Decoding of {Reed}-{Solomon} codes beyond the error-correction bound.
\newblock {\em Journal of Complexity}, 13(1):180--193, 1997.

\bibitem[Sud00]{Sudan:survey}
M.~Sudan.
\newblock List decoding: Algorithms and applications.
\newblock {\em SIGACT News}, 31(1):16--27, 2000.

\bibitem[Tre03]{luca-xor}
L.~Trevisan.
\newblock List-decoding using the {XOR} lemma.
\newblock In {\em Proc. $44^{th}$ {IEEE} Symposium on Foundations of Computer
  Science (FOCS'03)}, page 126, 2003.

\bibitem[TSZS01]{TZS}
A.~Ta-Shma, D.~Zuckerman, and S.~Safra.
\newblock Extractors from {R}eed-{M}uller codes.
\newblock In {\em Proc. $42^{nd}$ IEEE Symp. on Foundations of Computer Science
  (FOCS'01)}, pages 638--647, 2001.

\bibitem[TTV09]{trevisan2009regularity}
Luca Trevisan, Madhur Tulsiani, and Salil Vadhan.
\newblock Regularity, boosting, and efficiently simulating every high-entropy
  distribution.
\newblock In {\em Computational Complexity, 2009. CCC'09. 24th Annual IEEE
  Conference on}, pages 126--136. IEEE, 2009.

\bibitem[TZ11]{TZ}
T.~{Tao} and T.~{Ziegler}.
\newblock {The inverse conjecture for the Gowers norm over finite fields in low
  characteristic}.
\newblock {\em ArXiv e-prints}, January 2011.

\bibitem[Vad12]{Vad}
Salil~P. Vadhan.
\newblock Pseudorandomness.
\newblock {\em Foundations and Trends in Theoretical Computer Science},
  7(1-3):1--336, 2012.

\bibitem[Woz58]{Woz}
J.~Wozencraft.
\newblock List decoding.
\newblock Technical Report 48:90-95, Quarterly Progress Report, Research
  Laboratory of Electronics, MIT, 1958.

\bibitem[Zip79]{Zip}
R~E. Zippel.
\newblock Probabilistic algorithms for sparse polynomials.
\newblock {\em Proceedings of EUROSAM}, pages 216--226, 1979.

\end{thebibliography}

\end{document}